\documentclass[letterpaper,11pt]{article}

\usepackage{graphicx,epsf,psfrag,amsmath,amssymb}

\newtheorem{theorem}{Theorem}[section]
\newtheorem{lemma}{Lemma}[section]
\newtheorem{corollary}{Corollary}[section]
\newtheorem{definition}{Definition}[section]

\newcommand{\qed}{\hfill\hbox{\rlap{$\sqcap$}$\sqcup$}}
\newenvironment{proof}{\noindent \emph{Proof.\,}}{\qed}

\def\R{{\mathbb{R}}}

\def\etal{{\it et~al.}\,}
\newcommand{\frechet}{Fr\'echet}
\newcommand{\dfre}{d_{F}}

\usepackage{fancyhdr}
\newcommand{\myfootnote}[2]{%
	\thispagestyle{fancy}
	\fancyhf{}
	\renewcommand{\headrulewidth}{0pt}
	\lfoot{\footnotesize\emph{#1}}
	\rfoot{\footnotesize\emph{#2}}
}
\begin{document}

\begin{titlepage}  

\title{Voronoi Diagram of Polygonal Chains Under the Discrete Fr\'{e}chet Distance}
\author{
Sergey Bereg
\thanks{Department of Computer Science, University of Texas at Dallas, Richardson, TX 75083, USA. Email: {\tt besp@utdallas.edu}.}
\and
Marina Gavrilova
\thanks{Department of Computer Science, University of Calgary, Calgary, Alberta 
 T2N 1N4, Canada. Email: {\tt marina@cpsc.ucalgary.ca}.}
\and 
Binhai Zhu
\thanks{Department of Computer Science, Montana State University, Bozeman, MT 59717-3880, USA. Email: {\tt bhz@cs.montana.edu}.}
}

\date{}
\maketitle
\myfootnote{}{May 20, 2007}

\begin{abstract}

Polygonal chains are fundamental objects in many applications like
pattern recognition and protein structure alignment. A well-known
measure to characterize the similarity of two polygonal chains is
the famous Fr\'{e}chet distance. In this paper, for the first time,
we consider the Voronoi diagram of polygonal chains in $d$-dimension ($d=2,3$)
under the discrete Fr\'{e}chet distance. Given $n$ polygonal chains ${\cal C}$
in $d$-dimension ($d=2,3$), each with at most $k$ vertices, we prove
fundamental properties of such a Voronoi diagram {\em VD}$_F$(${\cal C}$).
Our main results are summarized as follows.
\begin{itemize}
\item The combinatorial complexity of {\em VD}$_F({\cal C})$ is at least
$\Omega(n^{\lfloor \frac{k+1}{2}\rfloor})$; in fact, even a slice of
{\em VD}$_F({\cal C})$ could contain a $L_\infty$-metric Voronoi diagram in
$k$ dimensions. 
\item The combinatorial complexity of {\em VD}$_F({\cal C})$ is at most
$O(n^{dk+\epsilon})$, for $d=2,3$.
\item Even if each polygonal chain degenerates into a point in three
dimensions (3D), e.g., when we are given a set $P$ of $n$ co-planar points
in 3D, the corresponding three-dimensional Voronoi diagram under the discrete
Fr\'{e}chet distance, {\em VD}$_F(P)$, has a combinatorial complexity of
$\Omega(n^2)$.
\end{itemize}

\end{abstract}

\end{titlepage}
\newpage

\section{Introduction}

Fr\'{e}chet distance was first defined by Maurice Fr\'{e}chet in 1906
\cite{Fr06}. While known as a famous distance measure in the field of
mathematics (more specifically, abstract spaces), it was Alt and Godau
who first applied it in measuring the similarity of polygonal curves
in 1992 \cite{AG92}. In general, Fr\'{e}chet distance between 2D
polygonal chains (polylines) can be computed in polynomial time \cite{AG92,AG95}, even
under translation or rotation (though the running time is much higher)
\cite{AKW01,We02}.
While computing (approximating) Fr\'{e}chet distance for surfaces is NP-hard
\cite{Go98}, it is polynomially solvable for restricted surfaces \cite{BBW06}.

In 1994, Eiter and Mannila defined the {\em discrete Fr\'{e}chet distance}
between two polygonal chains $A$ and $B$ and it turns out that this simplified
distance is always realized by two vertices in $A$ and $B$ \cite{EM94}.
They also showed that with dynamic programming the discrete Fr\'{e}chet
distance between them can be computed in $O(|A||B|)$ time. (In fact, the dynamic
programming algorithm works even when the vertices of the chains are in any
fixed dimension.) In \cite{In02}, Indyk defined a similar discrete Fr\'{e}chet
distance in some metric space and showed how to compute approximate nearest
neighbors using that distance.

Recently, Jiang, Xu and Zhu applied the discrete Fr\'{e}chet distance in
aligning the backbones of proteins (which is called the {\em protein
structure-structure alignment} problem) \cite{JXZ07}. In fact, in this
application the discrete Fr\'{e}chet distance makes more sense as the backbone
of a protein is simply a polygonal chain in 3D, with each vertex being the
alpha-carbon atom of a residue. So if the (continuous) Fr\'{e}chet distance
is realized by an alpha-carbon atom and some other point which does not
represent an atom, it is not meaningful biologically. Jiang, {\em et al.}
showed that given two planar polygonal chains the minimum discrete
Fr\'{e}chet distance between them, under both translation and rotation, can be
computed in polynomial time. They also applied some ideas therein to design an
efficient heuristic for the original protein structure-structure alignment
problem in 3D and the empirical results showed that their alignment is more
accurate compared with previously known solutions.

On the other hand, a lot is still unknown regarding discrete/continuous
Fr\'{e}chet distance. For instance, each Fr\'{e}chet distance is a true
distance measure, yet nobody has ever studied the Voronoi diagram
under such an important distance measure. This problem is fundamental, it has
potential applications, e.g., in protein structure alignment, especially with
the ever increasing computational power. Imagine that we have some polylines
$A_1,A_2,A_3...$ in space. If we can construct the Voronoi diagram for
$A_1,A_2,A_3,...$ in space, then given a new polyline $B$ we can easily compute
all the approximate alignment of $B$ with $A_i$'s. The movement of $B$ defines
a subspace (each point in the subspace represents a copy of $B$) and if we
sample this subspace evenly then all we need to do is to locate all these sample
points in the Voronoi diagram for $A_i$'s.

Unfortunately, nothing is known about Voronoi diagram under the
discrete/continuous Fr\'{e}chet distance, even for the simplest
case when each chain degenerates into a point in 3D. In this paper,
we will present the first set of such results by proving some fundamental
properties for both the general case and some special case. We believe that
these results will be essential for us to design efficient algorithms for
computing/approximating Voronoi diagram under the Fr\'{e}chet distance.
In this paper, all of our results are under the discrete Fr\'{e}chet distance unless otherwise specified.

The paper is organized as follows. In Section 2, we introduce some basic
definitions regarding \frechet\ distance and review some known results.
In Section 3, we show the combinatorial upper and lower bounds for
the Voronoi diagram of a set of 2D polylines ${\cal C}$,
{\em VD}$_{F}({\cal C})$. In Section 4, we sketch how to generalize the
results in Section 3 to a set of 3D polylines. In Section 5, we show the
quadratic lower bound for a special case where the $n$ polylines degenerate
into $n$ co-planar points in 3D. In Section 6, we conclude the paper with
several open problems.

\section{Preliminaries}

Given two polygonal chains $A,B$ with $|A|=k$ and $|B|=l$ vertices respectively,
we aim at measuring the similarity of $A$ and $B$ (possibly under translation
and rotation) such that their distance is minimized under certain measure.
Among the various distance measures, the Hausdorff distance is known to be
better suited for matching two point sets than for matching two polygonal
chains; the (continuous) \frechet\ distance is a superior measure for matching
two polygonal chains, but it is not quite easy to compute \cite{AG92}.

Let $X$ be the Euclidean plane $\mathbb{R}^2$; let $d(a,b)$ denote the
Euclidean distance between two points $a,b\in X$.
The (continuous) \frechet\ distance between two parametric curves $f:[0,1]\to X$
and $g:[0,1]\to X$ is 
$$
\delta_\mathcal{F}(f,g) = \inf_{\alpha,\beta} \max_{s\in[0,1]}
d(f(\alpha(s)),g(\beta(s))),
$$
where $\alpha$ and $\beta$ range over all continuous non-decreasing real
functions with $\alpha(0) = \beta(0) = 0$ and $\alpha(1) = \beta(1) = 1$
\footnote{This definition holds in any fixed dimensions. In Section 4, we will
cover the 3D case.}.

Imagine that a person and a dog walk along two different paths while connected
by a leash; they always move forward, though at different paces.
The minimum possible length of the leash is the \frechet\ distance between the
two paths.
To compute the \frechet\ distance between two polygonal curves $A$ and $B$
(in the Euclidean plane) of $|A|$ and $|B|$ vertices, respectively,
Alt and Godau \cite{AG92} presented an $O(|A||B|\log^2(|A||B|))$ time algorithm.
Later this bound was reduced to $O(|A||B|\log(|A||B|))$ time \cite{AG95}.

We now define the discrete \frechet\ distance following \cite{EM94}.

\begin{definition}
Given a polygonal chain (polyline) in the plane $P=$ $<p_1,\dots,p_k>$ of $k$
vertices, a \textbf{$m$-walk} along
$P$ partitions the path into $m$ disjoint non-empty subchains
$\{{\cal P}_i\}_{i=1..m}$ such that ${\cal P}_i=$ $<p_{k_{i-1}+1},\dots,p_{k_i}>$
and $0 = k_0 < k_1 < \dots < k_{m} = k$.

Given two planar polylines $A=$ $<a_1,\dots,a_k>$ and $B=$ $<b_1,\dots,b_l>$,
a \textbf{paired walk} along $A$ and $B$ is
a $m$-walk $\{{\cal A}_i\}_{i=1..m}$ along $A$ and
a $m$-walk $\{{\cal B}_i\}_{i=1..m}$ along $B$ for some $m$, such that,
for $1 \le i \le m$, $|{\cal A}_i| = 1$ or $|{\cal B}_i| = 1$
(that is, ${\cal A}_i$ or ${\cal B}_i$ contains exactly one vertex).
The \textbf{cost} of a paired walk
$W = \{({\cal A}_i,{\cal B}_i)\}$ along two paths $A$ and $B$ is
$$
\dfre^W(A,B) = \max_i \max_{(a,b) \in {\cal A}_i \times {\cal B}_i} d(a,b).
$$

The \textbf{discrete \frechet\ distance} between two polylines $A$ and $B$ is
$$
\dfre(A,B) = \min_W \dfre^W(A,B).
$$
The paired walk that achieves the discrete \frechet\ distance between two paths
$A$ and $B$ is also called the \textbf{\frechet\ alignment} of $A$ and $B$.
\end{definition}

Consider the scenario in which the person walks along $A$ and the dog along $B$.
Intuitively, the definition of the paired walk is based on three cases:
\begin{enumerate}
\item $|{\cal B}_i| > |{\cal A}_i| = 1$:
the person stays and the dog moves forward;
\item $|{\cal A}_i| > |{\cal B}_i| = 1$:
the person moves forward and the dog stays;
\item $|{\cal A}_i| = |{\cal B}_i| = 1$:
both the person and the dog move forward.
\end{enumerate}

\begin{figure}[hbt]
\centerline{\epsffile{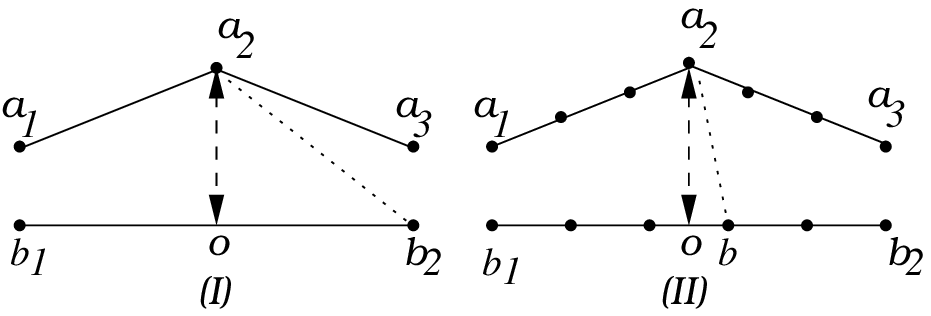}}
\begin{center}{\small {\bf Fig. 1}. The relationship between
discrete and continuous \frechet\ distances.}
\end{center}
\end{figure}

Eiter and Mannila presented a simple dynamic programming algorithm
to compute $\dfre(A,B)$ in $O(|A||B|)=O(kl)$ time \cite{EM94}. Recently,
Jiang \etal\ showed that the minimum discrete \frechet\ distance between $A$
and $B$ under translation can be computed in $O(k^3l^3\log(k+l))$ time, and
under both translation and rotation it can be computed in $O(k^4l^4\log(k+l))$
time \cite{JXZ07}. They are significantly faster than the corresponding bounds
for the continuous \frechet\ distance.
For continuous \frechet\ distance, under translation, the current fastest
algorithm for computing the minimum (continuous) \frechet\ distance between
$A,B$ takes $O((kl)^3(k+l)^2\log(k+l))$ time \cite{AKW01};
under both translation and rotation, the bound is $O((k+l)^{11}\log(k+l))$ time \cite{We02}.

We comment that while the discrete \frechet\ distance could be arbitrarily
larger than the corresponding continuous \frechet\ distance (e.g., in Fig.~1 (I),
they are $d(a_2,b_2)$ and $d(a_2,o)$ respectively), by adding sample points on the
polylines, one can easily obtain a close approximation of the continuous
\frechet\ distance using the discrete \frechet\ distance (e.g., one can
use $d(a_2,b)$ in Fig.~1 (II) to approximate $d(a_2,o)$). This fact was
also pointed out by Indyk in \cite{In02}. Moreover, the discrete \frechet\ distance
is a more natural measure for matching the geometric shapes of biological
sequences such as proteins. As we mentioned in the introduction, in such an
application, continuous \frechet\ does not make much sense to biologists.

In the remaining part of this paper, for the first time, we investigate the
Voronoi diagram of a set of polygonal chains (polylines) in $d$-dimension,
for $d=2,3$. While Voronoi diagram is a central structure in geometric
computing and has been widely
studied \cite{PS85}, it still attracts a lot of attention recently
\cite{AMT06,NBN07} (including a new annual International Symposium on Voronoi Diagrams
in Science and Engineering). We hope that our work will facilitate the
understanding of the continuous and discrete \frechet\ distance and further
enable their applications in various areas, like pattern recognition and
computational biology.

\section{The Combinatorial Complexity of {\em VD}$_{F}({\cal C})$ in 2D}

Let $A_k=$ $<a_1,a_2,\dots,a_k>$ and $B_l=$ $<b_1,b_2,\dots,b_l>$ be two 
polygonal chains in the plane where $a_i=(x(a_i),y(a_i)), b_j=(x(b_j),y(b_j))$
and $k,l\ge 1$. We first have the following lemma, which is easy to prove. 

\begin{lemma} \label{l1}
Let $A_2=$ $<a_1,a_2>$ and $B_2=$ $<b_1,b_2>$ be two line segments in the plane, then
$$\dfre(A_2,B_2)=\max(d(a_1,b_1),d(a_2,b_2)).$$
\end{lemma}

For general polylines, we can generalize the above lemma as follows.

\begin{lemma} \label{FFD}
Let $A_k=$ $<a_1,a_2,\dots,a_k>$ and $B_l=$ $<b_1,b_2,\dots,b_l>$ be two
polygonal chains in the plane where $k,l\ge 1$. The discrete
Fr\'{e}chet distance between $A_k$ and $B_l$ can be computed as
\begin{equation} \label{FD}
\dfre(A_k,B_l)= \begin{cases}
\max \{d(a_i,b_1),i=1,2,\dots,k\}& \text{if $l=1$,}\\
\max \{d(a_1,b_j), j=1,2,\dots,l\}& \text{if $k=1$,}\\
\max \{d(a_k,b_l),\min(\dfre(A_{k-1},B_{l-1}), \dfre(A_{k},B_{l-1}),\dfre(A_{k-1},B_{l})) \}& \text{if $k,l>1$.}
\end{cases}
\end{equation}
\end{lemma}

\begin{proof}
As can be seen from Section 2, the discrete Fr\'{e}chet distance can be computed
using dynamic programming. We assume that the man walks along $A_k$ and the
dog walks along $B_l$. If $l=1$ then the dog stays at $b_1$ and
$\dfre(A_k,B_l)=\max \{d(a_i,b_1),i=1,2,\dots,k\}$. If $k=1$ then the man
stays at $a_1$ and 
$\dfre(A_k,B_l)=\max \{d(a_1,b_j), j=1,2,\dots,l\}$. 

Suppose that both $k$ and $l$ are greater than 1. If the man and dog
both move in the last step then 
\begin{equation} \label{man-dog}
\dfre(A_k,B_l)=\max(d(a_k,b_l),\dfre(A_{k-1},B_{l-1})).
\end{equation}
If only the man moves in the last step then 
\begin{equation} \label{man-only}
\dfre(A_k,B_l)=\max(d(a_k,b_l),\dfre(A_{k-1},B_{l})).
\end{equation}
If only the dog moves in the last step then 
\begin{equation} \label{dog-only}
\dfre(A_k,B_l)=\max(d(a_k,b_l),\dfre(A_{k},B_{l-1})).
\end{equation}
The smallest value of (\ref{man-dog}),(\ref{man-only}) and
(\ref{dog-only}) is 
$$\max \{d(a_k,b_l),\min(\dfre(A_{k-1},B_{l-1}),
\dfre(A_{k},B_{l-1}),\dfre(A_{k-1},B_{l})) \}.$$
The lemma follows.
\end{proof}

Based on the above lemma, we try to investigate the combinatorial
complexity of {\em VD}$_{F}({\cal C})$, the Voronoi diagram of
a set ${\cal C}$ of $n$ planar polylines each with at most $k$ vertices.
Following \cite{ES86,Sh94}, a Voronoi diagram is a minimization of
distance functions to the sites (in this case the polylines in ${\cal C}$).
We first quickly review a result on the upper bound of lower envelopes in high
dimensions by Sharir \cite{Sh94}.

Let $\Sigma=\{\sigma_1,\dots,\sigma_n\}$ be a collection of $n$
$(d-1)$-dimensional algebraic surface patches in $d$-space. 
Let ${\cal A}(\Sigma)$ be the arrangement of $\Sigma$. The result
in \cite{Sh94} holds upon the following three conditions.

(i) Each $\sigma_i$ is monotone in the $x_1x_2\dots x_{d-1}$-direction
  (i.e. any line parallel to $x_d$-axis intersects $\sigma_i$ in at most
  one point). Moreover, each $\sigma_i$ is a portion of a
  $(d-1)$-dimensional algebraic surface of constant maximum degree
  $b$. 

(ii) The projection of $\sigma_i$ in $x_d$-direction onto the
  hyperplane $x_d=0$ is a semi-algebraic set defined in terms of a
  constant number of $(d-1)$-variate polynomials of constant maximum
  degree. 

(iii) The surface patches $\sigma_i$ are in {\em general position}
  meaning that the coefficients of the polynomials defining surfaces
  and their boundaries are algebraically independent over rationals. 

\begin{theorem}\cite{Sh94} \label{SH}
Let $\Sigma$ be a collection of $n$ $(d-1)$-dimensional algebraic
surface patches in $d$-space, which satisfy the above conditions
(i),(ii), and (iii). Then the number of vertices of ${\cal A}(\Sigma)$
that lie at the lower envelope (i.e., level one) is
$O(n^{d-1+\epsilon})$, for any $\epsilon>0$. 
\end{theorem}

We now try to show a general upper bound on the combinatorial complexity
of {\em VD}$_{F}({\cal C})$.

{\em Polyline-point correspondence}. 
Let $c_{1},c_{2},\dots,c_{k}$ be the sequence of vertices of
a polygonal chain $C$ and let $(x(c_{i}),y(c_{i}))$ be the
coordinates of vertex $c_{i},i=1,\dots,k$. 
Using this notation we view every polygonal chain $C$ with $k$
vertices in the plane as a point in $\R^{2k}$ and vice versa.

\begin{lemma} \label{surface}
Let $B\in\R^{2l}$ be a polygonal chain of $l$ vertices
$b_{1},\dots,b_{l}$ in the plane where
$b_{i}=(x(b_{i}),y(b_{i}))$. 
Let $f:\R^{2k}\to\R$ be the distance function defined as
$$f(C)=\dfre(C,B),$$
where $C=(c_1,\dots,c_k)\in\R^{2k}$ is a polygonal chain with $k$
vertices and $c_i=(x(c_i),y(c_i)),i=1,\dots,k$. 
The space $\R^{2k}$ can be partitioned into at most $(kl)!$
semi-algebraic sets $R_1,R_2,R_3,\dots$ such that the function 
$f(C)$ with domain restricted to any $R_i$ is algebraic.
Thus, the function $f(C)$ satisfies the above condition (i) and (ii). 
\end{lemma}

\begin{proof}
We consider ${k\choose 2}\cdot{l\choose 2}$ manifolds in $\R^{2k}$
defined as 
$$(x(c_i)-x(b_j))^2+(y(c_i)-y(b_{j}))^2=
(x(c_{i'})-x(b_{j'}))^2+(y(c_{i'})-y(b_{j'}))^2$$
for every four integer $i,i',j,j'$ such that $1\le i<i'\le k$ and
$1\le j<j'\le l$. They partition $\R^{2k}$ into at most $(kl)!$
semi-algebraic sets $R_1,R_2,R_3,\dots$ corresponding to the order of
distances between $c_i$ and $b_{j}$ for all $1\le i\le k$ and 
$1\le j\le l$.  

Equation (\ref{FD}) in Lemma \ref{FFD} implies that the function $f(C)$ restricted
to a domain $R_m$ satisfies 
$$f(C)=\sqrt{(x(c_i)-x(b_{j}))^2+(y(c_i)-y(b_{j}))^2}$$
for some pair $i,j$.
The lemma follows.
\end{proof}

We now prove the following theorem regarding the combinatorial
upper bound for {\em VD}$_{F}({\cal C})$.

\begin{theorem} \label{FVD}
Let ${\cal C}$ be a collection of $n$ polygonal chains $C_1,\dots,C_n$
each with at most $k$ vertices in the plane. 
The combinatorial complexity of the Voronoi diagram {\em VD}$_F({\cal C})$
is $O(n^{2k+\epsilon})$, for any $\epsilon>0$. 
\end{theorem}

\begin{proof}
Let $c_{i1},c_{i2},\dots,c_{ik}$ be the sequence of vertices of
the chain $C_i,i=1,\dots,n$ and let $(x(c_{ij}),y(c_{ij}))$ be the
coordinates of the vertex $c_{ij},j=1,\dots,k$.
Using this notation we view every polygonal chain $C$ with $k$
vertices in the plane as a point in $\R^{2k}$.
The Voronoi diagram {\em VD}$_F({\cal C})$ can be viewed as a diagram in
$\R^{2k}$. It is well-known that a Voronoi diagram can be interpreted as
a minimization of distance functions to the sites \cite{ES86,Sh94}.

Let $f_i:\R^{2k}\to\R,i=1,\dots,n,$ be the distance function defined as 
$$f_i(C)=\dfre(C,C_i),$$
where $C=(c_1,\dots,c_k)\in\R^{2k}$ is a polygonal chain with $k$
vertices and $c_i=(x(c_i),y(c_i)),i=1,\dots,k$. 
By Lemma \ref{surface} the function $f_i(C)$ satisfies the conditions
(i) and (ii). And condition (iii) is easily satisfied using an argument
similar to \cite{Sh94}. The Voronoi diagram {\em VD}$_F({\cal C})$ corresponds
to the lower envelope in the arrangement of the surfaces
$z=f_i(C)$ in $\R^{2k+1}$. By Theorem \ref{SH} the combinatorial complexity of
{\em VD}$_F({\cal C})$ is $O(n^{2k+\epsilon})$. 
\end{proof}

{\em Combinatorial lower bound} for {\em VD}$_F({\cal C})$. 
We now present a general lower bound for {\em VD}$_{F}({\cal C})$. In fact
we show a much stronger result that even a slice of {\em VD}$_F({\cal C})$
could contain a $L_\infty$-Voronoi diagram in $k$ dimensions, whose
combinatorial complexity is $\Omega(n^{\lfloor \frac{k+1}{2}\rfloor})$.

Schaudt and Drysdale proved that a $L_\infty$-Voronoi diagram in $k$
dimensions has combinatorial complexity of 
$\Omega(n^{\lfloor \frac{k+1}{2}\rfloor})$ \cite{SD92}.
Let $S=\{p_1,p_2,\dots,p_n\}$ be a set of $n$ points in $\R^k$ such
that the $L_\infty$ Voronoi diagram of $S$ has complexity of 
$\Omega(n^{\lfloor \frac{k+1}{2}\rfloor})$.
Let $M>0$ be a real number such that the hypercube $[-M,M]^k$ contains
$S$ and 
all the Voronoi vertices of the $L_\infty$-Voronoi diagram of $S$.
We consider a $k$-dimensional flat $F$ of $\R^{2k}$ defined as
$F=\{(a_1,M,a_2,2M,\dots,a_k,kM)~|~a_1,\dots,a_k\in\R\}$ and 
a projection $\pi:F\to\R^k$ defined as $\pi(b)=(b_1,b_3,\dots,b_{2k-1})$,
for $b=(b_1,b_2,b_3,...,b_{2k-1},b_{2k})$.

Let ${\cal C}=\{C_1,C_2,...,C_n\}$, each $C_i$ being a planar polygonal
chain with $k$ vertices. Let $C_i=$ $<c_{i1},c_{i2},...,c_{ik}>$ and
$c_{im}=(x(c_{im}),y(c_{im}))$, for $m=1,2,...,k$. We set
$c_{im}=(p_{im},mM)$, for 
$1\leq i\leq n,1\leq m\leq k$. Clearly, every $C_i\in F$.
With $C_i$ we associate a point $C'_i=\pi(C_i)$ in $\R^k$.
We show that the intersection of $F$ and a {\em VD}$_F({\cal C})$ 
has complexity of $\Omega(n^{\lfloor \frac{k+1}{2}\rfloor})$.

Consider a point $T\in F$ such that $T'=\pi(T)$ is a
$L_\infty$-Voronoi vertex of $S$ in $\R^k$. Then $T'\in [-M,M]^k$. 
At this point, the question is: what is the discrete Fr\'{e}chet distance
between $T$ and a chain $C_i$?  Note that $\dfre^W(T,C_i)<M$ if and
only if $W=W_0$ where $W_0=\{(t_m,c_{im})~|~m=1,\dots,k\}$. Therefore 
$\dfre(T,C_i)=\dfre^{W_0}(T,C_i)=\max\{|x(t_{1})-x(c_{i1})|,
|x(t_{2})-x(c_{i2})|,...,|x(t_{j})-x(c_{ij})|,...,|x(t_{k})-x(c_{ik})|\}$.
This is exactly the $L_\infty$-distance between $T'$ and $C'_i$, or
$d_F(T,C_i)=d_\infty(T',C'_i)$. 
Then the slice of {\em VD}$_F({\cal C})$ contains the
$L_\infty$-Voronoi diagram of $S$ in $k$ dimensions. 
We thus have the following theorem.

\begin{theorem}
The combinatorial complexity of {\em VD}$_F({\cal C})$ for a set ${\cal C}$ of
$n$ planar polygonal chains with $k$ vertices is
$\Omega(n^{\lfloor \frac{k+1}{2}\rfloor})$; in fact even a
$k$-dimensional slice of
{\em VD}$_F({\cal C})$ can have a combinatorial complexity
of $\Omega(n^{\lfloor \frac{k+1}{2}\rfloor})$.
\end{theorem}

We comment that this lower bound is probably not tight and some significantly
different method is needed to improve either the lower or upper bounds
(or both). Nevertheless, the above theorems show the first non-trivial bounds
for the size of {\em VD}$_F({\cal C})$.

\section{The Combinatorial Complexity of {\em VD}$_{F}({\cal C})$ in 3D}

For protein-related applications, the input are polygonal chains in 3D.
It turns out that Lemmas \ref{l1} and \ref{FFD} hold for polygonal chains
in $\R^3$. Similar to Lemma \ref{surface} we can prove 

\begin{lemma} \label{surface3}
Let $B\in\R^{3l}$ be a polygonal chain of $l$ vertices
$b_{1},\dots,b_{l}$ in $\R^3$, where\\
$b_{i}=(x(b_{i}),y(b_{i}),z(b_{i}))$. 
Let $f:\R^{3k}\to\R$ be the distance function defined as
$$f(C)=\dfre(C,B),$$
where $C=(c_1,\dots,c_k)\in\R^{3k}$ is a polygonal chain with $k$
vertices and $c_i=(x(c_i),y(c_i),z(c_i)),i=1,\dots,k$. 
The space $\R^{3k}$ can be partitioned into at most $(kl)!$
semi-algebraic sets $R_1,R_2,R_3,\dots$ such that the function 
$f(C)$ with domain restricted to any $R_i$ is algebraic.
Thus, the function $f(C)$ satisfies the above condition (i) and (ii). 
\end{lemma}

The main difference is that the manifolds in $\R^{3k}$ are defined as 
$$(x(c_i)-x(b_j))^2+(y(c_i)-y(b_{j}))^2+(z(c_i)-z(b_{j}))^2=
(x(c_{i'})-x(b_{j'}))^2+(y(c_{i'})-y(b_{j'}))^2+(z(c_{i'})-z(b_{j'}))^2.$$

Using Lemma \ref{surface3} we can prove similar to Theorem \ref{FVD} the
following bound.

\begin{theorem} \label{FVD3}
Let ${\cal C}$ be a collection of $n$ polygonal chains $C_1,\dots,C_n$
each with at most $k$ vertices in $\R^3$. 
The combinatorial complexity of the Voronoi diagram {\em VD}$_F({\cal C})$
is $O(n^{3k+\epsilon})$, for any $\epsilon>0$. 
\end{theorem}

We comment that the lower bound in Theorem 3.3 is in fact a special case
for {\em VD}$_F({\cal C})$ in $\R^3$.

\section{Lower Bound for a Simple Degenerate Case}

In this section, we present a lower bound for a simple degenerate case.
Under this degenerate case, the geometric bisectors are easy to visualize.
Nevertheless, we will see that Voronoi diagram under the discrete
(and continuous) Fr\'{e}chet distance is much more complex compared with the
corresponding Euclidean case. In fact, for the same construction the former
Voronoi diagram has a combinatorial complexity of $\Omega(n^2)$ while the
latter Voronoi diagram (under the Euclidean metric) is only of size $\Theta(n)$.

First, let all the polygonal chains be 2D line segments, i.e.,
${\cal C}=\{C_i|C_{i}=$ $<a_{i},b_{i}>, a_i=(a_{i1},a_{i2}), b_i=(b_{i1},b_{i2})$,
for $1\leq i \leq n\}$. Then we set $b_{i2}=0$ for all $i$. As a matter of fact,
${\cal C}$ is equivalent to a set $P$ of $n$ points in 3D, i.e.,
$P=\{(a_{i1},a_{i2},b_{i1})|1\leq i\leq n\}$.

We now further construct points in $P$ in two equal parts, i.e.,
$P=P_1\cup Q_1$, with $P_1=\{p_i=(0,0,i)|1\leq i\leq n/2 \}$ and
$Q_1=\{(j,0,0)|1\leq j\leq n/2 \}$ (without loss of generality, assume that $n$ is even).

Let us look at the bisector between $p_i$ and $p_{i+1}$. Be reminded that
$p_i$ and $p_{i+1}$ are in fact line segments; i.e., $p_i=$ $<(0,0),(i,0)>$ and
$p_{i+1}=$ $<(0,0), (i+1,0)>$. Let $s=(x,y,z,0)$ be a point on the bisector of
$p_i$ and $p_{i+1}$ in 3D. Then the discrete Fr\'{e}chet distance
$d_F(p_i,s)$ is $d_F(p_i,s)=\max\{|z-i|,\sqrt{x^2+y^2}\}$. Similarly,
$d_F(p_{i+1},s)=\max\{|z-i-1|,\sqrt{x^2+y^2}\}$. 

So, what does this bisector look like? When $|z-i|$ and $|z-i-1|$ are bigger
than $\sqrt{x^2+y^2}$, then the bisector is a plane $z=\frac{2i+1}{2}$; when
$\sqrt{x^2+y^2}$ is bigger the bisector are two paraboloids $(z-i)^2=x^2+y^2$ and
$(z-i-1)^2=x^2+y^2$. The intersection of the plane $z=\frac{2i+1}{2}$ and
these two paraboloids is, of course, a circle $x^2+y^2=(\frac{1}{2})^2$ (on the
plane $z=\frac{2i+1}{2}$). See Fig.~2 (part I).

\begin{figure}[hbt]
\centerline{\epsffile{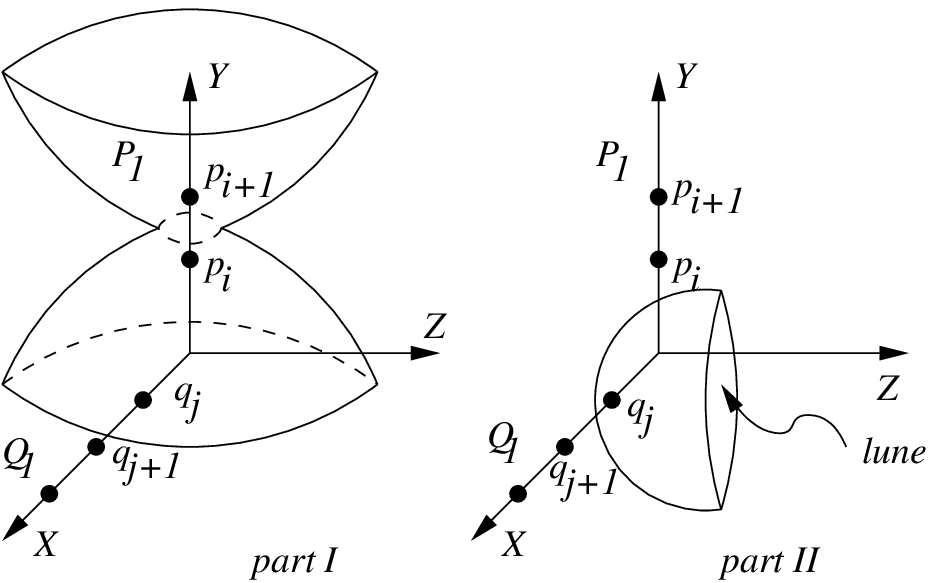}}
\begin{center}{\small {\bf Fig.~2}. Bisectors for $p_i,p_{i+1}$ and for $q_j,q_{j+1}$.}
\end{center}
\end{figure}

Next let us look at the bisector between $q_j$ and $q_{j+1}$. Again, be
reminded that $q_j$ and $q_{j+1}$ are in fact line segments; i.e., $q_j=$ $<(j,0),(0,0)>$ and
$q_{j+1}=$ $<(j+1,0),(0,0)>$. For a point $s=(x,y,z)$ on the bisector of
$q_j$ and $q_{j+1}$ in 3D, we have the discrete Fr\'{e}chet distance
$d_F(q_{j},s)=\max\{|z|,\sqrt{(x-j)^2+y^2}\}$ and
$d_F(q_{j+1},s)=\max\{|z|,\sqrt{(x-j-1)^2+y^2}\}$.

In this case, the bisector is even more interesting (and complex): When $|z|$ is
bigger, the bisector contains all the points satisfying 
$z^2 \geq (x-j)^2+y^2$ and $z^2 \geq (x-j-1)^2+y^2$ and this is
not even a manifold. Geometrically, this is the intersection of the two
paraboloids $z^2 \geq (x-j)^2+y^2$ and $z^2 \geq (x-j-1)^2+y^2$ and on
each vertical slice $z=c$, for some constant $c$, this is a lune determined by
two disks $(x-j)^2+y^2 \leq c^2$ and $(x-j-1)^2+y^2 \leq c^2$. Of course,
when $|z|$ is smaller, the bisector is exactly the plane $x=\frac{2j+1}{2}$.
See Fig.~2 (part II).

So how do we determine the combinatorial complexity of
{\em VD}$_{F}(P_1\cup Q_1)$? In this case, the Voronoi vertex is not a
geometric point. We prove below that we can take the first intersection
of the bisector of $p_i,p_{i+1}$ and the bisector of $q_j,q_{j+1}$, which is a
geometric point, and then show that this intersection point is farther to
some points in $P_1-\{p_i,p_{i+1}\}$ and $Q_1-\{q_j,q_{j+1}\}$ than to
$p_i,p_{i+1},q_j,q_{j+1}$ and it is never closer to any point
in $P_1-\{p_i,p_{i+1}\}$ and $Q_1-\{q_j,q_{j+1}\}$ than to
$p_i,p_{i+1},q_j,q_{j+1}$. Then this point contributes an $\Omega(1)$ cost to
the combinatorial complexity of {\em VD}$_{F}(P_1\cup Q_1)$. Finally we show
that there are $\Omega(n)$ such $p_i,p_{i+1}$ pairs and $\Omega(n)$ such
$q_j,q_{j+1}$ pairs. And this will conclude the proof that {\em VD}$_F(P)$
has a combinatorial complexity of $\Omega(n^2)$, or, put it in another way,
given $n$ co-planar points in 3D the Voronoi diagram under the discrete
Fr\'{e}chet distance has a combinatorial complexity of $\Omega(n^2)$. This
is significantly different from the corresponding Euclidean Voronoi diagram
of $n$ co-planar points, has a combinatorial complexity of $\Theta(n)$
\cite{PS85}.
 
Now let us finish the technical details. The first intersection
of the bisector of $p_i,p_{i+1}$ and the bisector of $q_j,q_{j+1}$ is determined
by $z^2 = (x-j-1)^2+y^2$ and $(z-i)^2=x^2+y^2$. Therefore, we have the
intersection point $t_{ij}=(\frac{j+1}{2},\frac{\sqrt{i^2-(j+1)^2}}{2},\frac{i}{2})$. 
Clearly, if two pairs $(i,j)$ and $(i',j')$ are different then
$t_{ij}\neq t_{i'j'}$.
We now show that $t_{ij}$ is never closer to points in $P_1-\{p_i,p_{i+1}\}$ and $Q_1-\{q_j,q_{j+1}\}$ than to
$p_i,p_{i+1},q_j,q_{j+1}$
(and $t_{ij}$ is farther to some points in $P_1-\{p_i,p_{i+1}\}$ and $Q_1-\{q_j,q_{j+1}\}$ than to
$p_i,p_{i+1},q_j,q_{j+1}$). This can be done by showing $d_F(t_{ij},p_h)>\frac{i}{2}$ for $h>i$ and
$d_F(t_{ij},p_h)=\frac{i}{2}$ for $h<i$;
and $d_F(t_{ij},q_m)>\frac{i}{2}$ for $m>j+1$ and
$d_F(t_{ij},q_k)=\frac{i}{2}$ for $m<j+1$. We have
$$d_F(t_{ij},p_h)=\max\{\sqrt{(\frac{j+1}{2})^2+\frac{i^2-(j+1)^2}{4}},|\frac{i}{2}-h|\},$$
which is $d_F(t_{ij},p_h)=\max\{\frac{i}{2},|\frac{i}{2}-h|\}>\frac{i}{2}$, when
$h>i$ and
$d_F(t_{ij},p_h)=\max\{\frac{i}{2},|\frac{i}{2}-h|\}=\frac{i}{2}$, when
$h<i$.
Similarly,
$$ d_F(t_{ij},q_m)=\max\{\sqrt{(\frac{j+1}{2}-m)^2+\frac{i^2-(j+1)^2}{4}},|\frac{i}{2}|\},$$
which is $d_F(t_{ij},q_m)=\max\{\frac{i^2-4m(j+1)+4m^2}{2},\frac{i}{2}\}>\frac{i}{2}$, when $m>j+1$ and\\
$d_F(t_{ij},q_m)=\max\{\frac{i^2-4m(j+1)+4m^2}{2},\frac{i}{2}\}>\frac{i}{2}$, when $m<j+1$.
Be reminded that if two pairs $(i,j)$ and $(i',j')$ are different then
$t_{ij}\neq t_{i'j'}$.
Consequently, each bisector of $p_i,p_{i+1}$ and each bisector of $q_j,q_{j+1}$
determines a distinct Voronoi point (on a Voronoi vertex) which is never closer
to points in $P_1-\{p_i,p_{i+1}\}$ and $Q_1-\{q_j,q_{j+1}\}$ than to
$p_i,p_{i+1},q_j,q_{j+1}$ and
is farther to some points in
$P_1-\{p_i,p_{i+1}\}$ and $Q_1-\{q_j,q_{j+1}\}$ than to
$p_i,p_{i+1},q_j,q_{j+1}$. We thus have the following theorem.

\begin{theorem}
The combinatorial complexity of {\em VD}$_F(P)$ for a set $P$ of $n$ degenerate
planar line segments, which is a set of $n$ co-planar points in 3D, is
$\Omega(n^2)$.
\end{theorem}

Since the continuous and the discrete Fr\'{e}chet distance are identical
for line segments, we also have the following corollary.

\begin{corollary}
The combinatorial complexity of the Voronoi diagram for a set $P$ of $n$
degenerate planar line segments (which is a set of $n$ co-planar points in 3D)
under the continuous Fr\'{e}chet distance measure is $\Omega(n^2)$.
\end{corollary}

\section{Concluding Remarks}

In this paper, for the first time, we study the Voronoi diagram
of polylines in 2D and 3D under the discrete \frechet\ distance. We show
combinatorial upper and lower bounds for such a Voronoi diagram.
At this point, to make the diagram useful, we need to present
efficient (approximation) algorithms to construct such a Voronoi diagram for
decent $k$ (say $k=10\sim20$) so that one can first use a $(k-1)$-link chain
to approximate a general input polyline. Although the running time might
still look too high (due to the $O(k)$ exponent in the running time), the
good news is that in many applications like protein structural alignment,
$n$ is not very large. Another open problem, obviously, is to improve
the combinatorial bounds shown in this paper.

\section*{Acknowledgment}
We thank Minghui Jiang for pointing out the gap in an earlier version of
the arguments for Theorem 3.3.

\end{document}